\documentclass[12pt]{article}
\usepackage{amsfonts,amsthm,color}
\newtheorem{theorem}{\indent Theorem}[section]

\newtheorem{lemma}[theorem]{\indent Lemma}

\newtheorem{assumption}[theorem]{\indent Assumption}
\def\a{\mathbf{a}}
\def\n{\mathbf{n}}

\def\R{\mathbb{R}}

\def\Re{\mathop{\rm Re}\nolimits}
\def\Im{\mathop{\rm Im}\nolimits}

\def\curl{\mathop{\rm curl}\nolimits}
\def\sgn{\mathop{\rm sgn}\nolimits}

\def\dist{\mathop{\rm dist}\nolimits}
\def\diam{\mathop{\rm diam}\nolimits}

\def\tr{\mathop{\rm tr}\nolimits}
\def\min{\mathop{\rm min}\limits}

\usepackage{amsmath}	
\begin{document}
\begin{center}
 {
\Large\bf
Schr\"odinger operators with random 
$\delta$ magnetic fields
}
\vspace{0.2cm}

by Takuya MINE\footnote{
Faculty of Arts and Sciences,
Kyoto Institute of Technology,
Matsugasaki,
Sakyo-ku,
Kyoto 606-8585,
Japan. E-mail: mine@kit.ac.jp}
and
Yuji NOMURA\footnote{
Graduate School of Material Science,
University of Hyogo, Shosha, Himeji 671-2280, Japan.
E-mail: y.nomura@sci.u-hyogo.ac.jp}
\vspace{0.2cm}
\end{center}

{\bf Abstract.} 
We shall consider the Schr\"odinger 
operators on $\R^2$ with 
random $\delta$ magnetic fields.
Under some mild conditions on the positions 
and the fluxes of the $\delta$-fields,
we prove the spectrum coincides with $[0,\infty)$
and the integrated density of states (IDS)
decays exponentially at the bottom of the spectrum
(Lifshitz tail),
by using the Hardy type inequality
by Laptev-Weidl \cite{L-W}.
We also give a lower bound for IDS at the bottom of the spectrum.

\section{Introduction}
We consider the Schr\"odinger operators on $\R^2$
with random magnetic fields
\begin{displaymath}
 {\cal L}_\omega = \left(\frac{1}{i}\nabla - \a_\omega\right)^2,
\end{displaymath}
where $\omega$ is an element from some probability space
$(\Omega, \mathbf{P})$
and the vector-valued function
$\a_\omega(x) = (a_{\omega,1}(x), a_{\omega,2}(x))$ 
($x=(x_1,x_2)\in \mathbb{R}^2$) is 
the magnetic vector potential dependent on $\omega$.
The magnetic field corresponding to $\a_\omega$
is given by
$\curl \a_\omega 
= \partial_1 a_{\omega,2}
 - \partial_2 a_{\omega,1}$
$(\partial_j=\partial/\partial x_j)$
and we assume
\begin{equation}
 \curl \a_\omega
 = \sum_{\gamma\in \Gamma_\omega}
   2\pi \alpha_\gamma(\omega) \delta_\gamma
 \label{1-1}
\end{equation}
in the distribution sense,
where $\Gamma_\omega$ is a discrete set in $\mathbb{R}^2$
dependent on $\omega$ without accumulation points in $\mathbb{R}^2$,
$\alpha(\omega) = \{\alpha_\gamma(\omega)\}_{\gamma\in \Gamma_\omega}$
is a sequence of real numbers satisfying $0 \leq \alpha_\gamma(\omega) <1$
and dependent on $\omega$,
and $\delta_\gamma$ is the Dirac measure supported
on the point $\gamma$.
For any given $(\Gamma_\omega, \alpha(\omega))$,
we can construct vector potential 
$\a_\omega\in C^\infty(\mathbb{R}^2\setminus\Gamma_\omega;\mathbb{R}^2)$ 
satisfying (\ref{1-1}) (see (\ref{1-2}) below).
The assumption $0 \leq \alpha_\gamma(\omega) <1$ loses no generality,
because we can shift the value of $\alpha_\gamma$ by any integer value
by using suitable gauge transform (see Lemma \ref{gauge} below).

Before stating our assumptions,
we prepare some notation used in the present paper.
For $S\subset \mathbb{R}^2$, $x\in \mathbb{R}^2$,
and $r>0$,
let
$S+x = \{s + x\mid s\in S\}$ and
$r S = \{r s \mid s\in S\}$.
For $k\geq 0$, let 
\begin{displaymath}
Q_k = 
\left\{ (x_1,x_2)\in \mathbb{R}^2\mid
 -k - \frac{1}{2} \leq x_j < k+\frac{1}{2}\ (j=1,2)\right\},
\end{displaymath}
which is a square with edge length $2k+1$ 
centered at the origin.
Especially $Q_0$ is a unit square centered at the origin.
The boundary of a set $S$ is denoted by $\partial S$.
The open disc of radius $r$ centered at $x$
is denoted by $B_{x}(r)$, that is,
\begin{displaymath}
B_{x}(r)=\{y\in \mathbb{R}^2\mid |y-x|<r\}. 
\end{displaymath}

Our assumption
is as follows.
\begin{assumption}
\label{assumptions}
Let $(\Omega, \mathbf{P})$ be a probability space,
$\Gamma_\omega$ a discrete set in $\mathbb{R}^2$
dependent on $\omega\in \Omega$ without accumulation points in $\mathbb{R}^2$,
and $\alpha(\omega)=\{\alpha_\gamma(\omega)\}_{\gamma\in \Gamma_\omega}$
a sequence of real numbers with $0\leq \alpha_\gamma(\omega)<1$
dependent on $\omega\in \Omega$.
For a Borel set $E$ in $\R^2$, put
\begin{displaymath}
  \Phi_\omega(E) 
 = \sum_{\gamma\in \Gamma_\omega \cap E} \alpha_\gamma(\omega).
\end{displaymath}
We assume the following conditions (i)-(vi).

 \begin{enumerate}
 \item For any Borel set $E$ in $\R^2$,
 the random variable $\Phi(E):\omega\mapsto\Phi_\omega(E)$
is measurable with respect to $\omega \in \Omega$.
 \item
 For any finite distinct points
$\{n_j\}_{j=1}^J$ with $n_j\in \mathbb{Z}^2$,
and for any Borel sets $\{E_j\}_{j=1}^J$ with $E_j\subset n_j + Q_0$,
 the random variables
 $\{\Phi(E_j)\}_{j=1}^J$
are independent.

  \item For any Borel set $E\subset Q_0$,
 the random variables
 $\{\Phi(E+n)\}_{n\in \mathbb{Z}^2}$
 are identically distributed.

 \item 
The mathematical expectation $\mathbf{E}[\Phi(Q_0)]$
is positive and finite.
The variance $\mathbf{V}[\Phi(Q_0)]$ is finite.

\item 
$\Phi_\omega(\partial Q_0)=0$ almost surely.

\item 
One of the following two conditions (a) or (b) holds.
\begin{enumerate}
 \item 
There exists a positive constant $c$
with $0<c\leq 1$ independent of $\omega$ such that
the probability of the event 
`the following two conditions 
(\ref{small_flux}) and (\ref{separated}) simultaneously hold'
is positive for any $\epsilon>0$.
\begin{eqnarray}
\label{small_flux}
&&\Phi_\omega(Q_0)=\sum_{\gamma\in \Gamma_\omega\cap Q_0}\alpha_\gamma<
\epsilon,\\
&&B_\gamma(c\sqrt{\alpha_\gamma})\cap
B_{\gamma'}(c\sqrt{\alpha_{\gamma'}})=\emptyset,\quad
B_\gamma(c\sqrt{\alpha_\gamma})\cap \partial Q_0 =\emptyset\nonumber\\
&&
\quad \mbox{for every }\gamma,\gamma'\in \Gamma_\omega \cap Q_0
\mbox{ with }\gamma\not=\gamma'.
\label{separated}
\end{eqnarray}

 \item 
The probability of the event
\begin{equation}
\label{small_square_root_flux}
\sum_{\gamma\in \Gamma_\omega\cap Q_0}\sqrt{\alpha_\gamma}< \epsilon 
\end{equation}
is positive for any $\epsilon>0$.
\end{enumerate}
\end{enumerate}
\end{assumption}
\noindent
In a part of our main result,
we assume a stronger condition as follows. 
\begin{assumption}
\label{assumption2}
In addition to Assumption \ref{assumptions},
there exist positive constants 
$c_1$ and $\delta_1$ such that
\begin{eqnarray}
 \label{assumption2_1}
&&\mathbf{P}\left\{
\mbox{
 (\ref{small_flux}) and (\ref{separated}) hold
}
\right\}\geq c_1 \epsilon^{\delta_1}\quad
 \mbox{(if (vi)(a) holds)},\\
 \label{assumption2_2}
&&\mathbf{P}\left\{
\mbox{
 (\ref{small_square_root_flux}) holds
}
\right\}\geq c_1 \epsilon^{\delta_1}\quad
 \mbox{(if (vi)(b) holds)}
\end{eqnarray}
for sufficiently small $\epsilon>0$,
where $\epsilon$ is the one in (\ref{small_flux}) 
or (\ref{small_square_root_flux}), respectively.
\end{assumption}

\noindent
The assumption (\ref{small_flux}) means 
the magnetic flux through $Q_0$ can be arbitrarily small,
and (\ref{separated}) means the points $\Gamma_\omega$
are separated farther than 
a constant multiple of 
the magnetic length $\sqrt{\alpha_\gamma}$
as the flux tends to $0$.
The assumption (\ref{small_square_root_flux}) 
is independent of the positions of the points $\Gamma_\omega$,
but the restriction on the flux is stronger than (\ref{small_flux}),
since $0\leq \alpha_\gamma\leq\sqrt{\alpha_\gamma}\leq 1$.
If the number of $\Gamma_\omega\cap Q_0$ is bounded
by a constant independent of $\omega$,
then (\ref{small_flux}) implies (\ref{small_square_root_flux})
by the Schwarz inequality.
These conditions guarantee the spectrum 
of our Hamiltonian is $[0,\infty)$ (Theorem \ref{theorem_spectrum}).

There are numerous examples satisfying 
Assumption \ref{assumptions} or \ref{assumption2}.
We list some typical examples below.
\begin{enumerate}
 \item {\bf Perturbation of a lattice.}
Let $\Gamma_\omega = \{n + f_n(\omega)\}
_{n\in \mathbb{Z}^2}$,
where $\{f_n\}$
are independently, identically distributed (\textit{abbrev.}\ i.i.d.)
$\mathbb{R}^2$-valued random variables with values in 
the interior of $Q_0$.
The fluxes $\{\alpha_\gamma\}$ are $[0,1)$-valued i.i.d.\ random variables
independent of $\{f_n\}$, satisfying
$\mathbf{E}[\alpha_\gamma]>0$ and
\begin{displaymath}
\mathbf{P}\{\alpha_\gamma < \epsilon\}>0\quad
\mbox{for any }\epsilon>0.
\end{displaymath}
Then Assumption \ref{assumptions} is satisfied.
Moreover, if additionally
\begin{displaymath}
\mathbf{P}\{\alpha_\gamma < \epsilon\}\geq c_1 \epsilon^{\delta_1}\quad
\mbox{for sufficiently small }\epsilon>0
\end{displaymath}
for some positive constants $c_1$ and $\delta_1$,
then Assumption \ref{assumption2} is satisfied.

 \item {\bf Poisson model.}
The random set $\Gamma_\omega$ is the Poisson configuration
with intensity measure $\rho\, dx$,
where $\rho$ is some positive constant,
i.e.\ the following holds (see e.g. \cite{Rei, A-I-K-N}).
\begin{enumerate}
 \item For any Borel set $E$ with finite Lebesgue measure $|E|$,
\begin{displaymath}
 \mathbf{P}\left\{{\#(\Gamma_\omega \cap E)}=j\right\}=e^{-\rho|E|}\frac{(\rho|E|)^j}{j!}
\quad (j=0,1,2,\ldots){,}
\end{displaymath}
{where $\# S$ denotes the number of the points in the set $S$.}
 \item For any disjoint Borel sets $E_1,\ldots,E_n$ with finite
Lebesgue measure,
the random variables $\{{\#(\Gamma_\omega \cap E_j)}\}_{j=1}^n$ are independent.
\end{enumerate}
The fluxes $\{\alpha_\gamma\}$ are i.i.d.\ random variables
independent of $\Gamma_\omega$ and satisfying
$\mathbf{E}[\alpha_\gamma]>0$
(${\alpha_\gamma}$ can be a constant sequence).
Then Assumption \ref{assumption2} holds, since
\begin{displaymath}
 \mathbf{P}\{\Phi_\omega(Q_0)=0\}{\geq} e^{-\rho}>0.
\end{displaymath}

 \item {\bf Accumulating lattice.}
This is somewhat an artificial example which satisfies
(vi)(a) of Assumption \ref{assumptions}
but does not satisfy (vi)(b).
Considering the i.i.d.\ assumption 
((ii) and (iii) of Assumption \ref{assumptions}),
we give the distribution of $\Gamma_\omega$ only in $Q_0$
and the distribution of  $\alpha_\gamma(\omega)$ for $\gamma \in \Gamma_\omega\cap Q_0$.
Let $\Omega_0$ be the set of the positive integers
with the probability measure $\mathbf{P}\{m\}=6/(m\pi)^2$
for $m \in \Omega_0$.
For $m\in \Omega_0$, we define
\begin{eqnarray*}
&& \Gamma_m
= \left\{\left(
\frac{n_1}{2m+1},\frac{n_2}{2m+1}\right)
\mid n_j\in \mathbb{Z},\ |n_j|\leq m\ (j=1,2)
\right\},\\
&& \alpha_\gamma(m) = \frac{1}{(2m+1)^3}\quad
(\gamma \in \Gamma_m).
\end{eqnarray*}
Then, we have
$\Phi_m(Q_0)=(2m+1)^{-1}$ and
\begin{displaymath}
 \mathbf{P}
\left\{\Phi_m(Q_0)<\epsilon\right\}
=
\frac{6}{\pi^2}\sum_{m>(\epsilon^{-1}-1)/2} \frac{1}{m^2}>
c_1 \epsilon
\end{displaymath}
for some positive constant $c_1$ and sufficiently small $\epsilon>0$.
Moreover, since $\sqrt{\alpha_\gamma(m)} =(2m+1)^{-3/2}$ and
$\min_{\gamma\not=\gamma'}|\gamma-\gamma'|=(2m+1)^{-1}$, 
(\ref{separated}) always holds if we take $c=1/2$.
Thus, for small $\epsilon>0$,
(\ref{assumption2_1}) holds with $\delta_1=1$.
But
\begin{displaymath}
 \sum_{\gamma\in \Gamma_m} \sqrt{\alpha_\gamma}
=(2m+1)^{1/2}\geq \sqrt{3}
\end{displaymath}
for every $m\in \Omega_0$, so
the probability of the event 
(\ref{small_square_root_flux}) is $0$
for any $0<\epsilon<\sqrt{3}$.
\end{enumerate}
An example satisfying (vi)(b) but not satisfying (vi)(a)
can be more easily constructed,
by considering the two-point fields approaching very fast
to each other as the fluxes tend to $0$.

According to \cite[section 4]{G-S},
we can construct the vector potential $\a_\omega$
satisfying (\ref{1-1}) as follows.
For notational convenience,
we identify $x=(x_1,x_2)\in \mathbb{R}^2$ with
$z=x_1+ix_2\in \mathbb{C}$.
If a meromorphic function
$\phi_\omega(z)$ has poles only on $\Gamma_\omega$
and the principal part  of $\phi_\omega$
at $z=\gamma$ is $\alpha_\gamma/(z-\gamma)$,
then the Cauchy-Riemann relation 
and the distributional equality $\Delta \log|z-\gamma|=2\pi \delta_\gamma$
imply the vector potential
\begin{equation}
\label{1-2}
 \a_\omega = (\Im \phi_\omega, \Re \phi_\omega)
\end{equation}
satisfies (\ref{1-1}).
Such a meromorphic function $\phi_\omega$
always exists by the Mittag--Leffler theorem.
Under Assumption \ref{assumptions},
the function $\phi_\omega$ is explicitly given by
\begin{displaymath}
 \phi_\omega (z)= \frac{\alpha_0(\omega)}{z} + 
               \sum_{\gamma\in \Gamma_\omega\setminus\{0\}}
     \alpha_\gamma (\omega)
     \left(\frac{1}{z-\gamma } + \frac{1}{\gamma}
          + \frac{z}{\gamma^2} \right),
\end{displaymath}
where we put $\alpha_0(\omega)=0$ if $0 \not\in\Gamma$.
We can prove that the sum in the above formula
converges almost surely
under Assumption \ref{assumptions}
(a similar argument is seen in \cite[Proposition 4.1]{M-N}).

We denote the Friedrichs extension of the operator
${\cal L}_\omega$ with the operator domain $C_0^\infty(\R^2\setminus \Gamma_\omega)$
by $H_\omega$,
where $C_0^\infty(U)$ denotes the set of 
the compactly supported smooth functions whose support is contained in $U$.
The operator $H_\omega$ is a non-negative 
self-adjoint operator on $L^2(\mathbb{R}^2)$,
and the operator domain $D(H_\omega)$ of $H_\omega$ 
is given by
\begin{eqnarray}
&&D(H_\omega)=\{u \in L^2(\mathbb{R}^2)\mid 
 {\cal L}_\omega u \in L^2(\R^2),\nonumber\\
&&\hspace{4.2cm} \limsup_{x\rightarrow \gamma} |u(x)|<\infty\quad
 \mbox{for any }\gamma\in \Gamma_\omega
\},
 \label{bdry}
\end{eqnarray}
where the derivative ${\cal L}_\omega u$ is interpreted 
in the sense of the Schwartz distribution on 
$\mathbb{R}^2\setminus\Gamma_\omega$.

For the spectrum $\sigma(H_\omega)$
of $H_\omega$, we obtain the following result.
\begin{theorem}
 \label{theorem_spectrum}
Under Assumption \ref{assumptions},
we have $\sigma(H_\omega) = [0,\infty)$ almost surely.
\end{theorem}

There are numerous results similar to Theorem \ref{theorem_spectrum}
in the theory of random Schr\"odinger operators
(see e.g.\ \cite{C-L, S1, A-I-K-N, Ki-Ma, Ki1}).
Especially,
Borg \cite[Theorem 4.3.1]{B} proves the special case of 
Theorem \ref{theorem_spectrum},
in the case $\Gamma_\omega$ is a non-random lattice.
Nevertheless,
the proof of Theorem \ref{theorem_spectrum} is not trivial
from the following reason.

The main strategy to prove the almost sure spectrum $\Sigma= [0,\infty)$
used in the known results is as follows.
First, we find a class of the \textit{admissible operators} ${\cal A}$,
such that $\Sigma$ is expressed as
the closure of the union of the spectrum of 
all the operators belonging to ${\cal A}$
(see e.g.\ Kirsch-Martinelli \cite[Theorem 3]{Ki-Ma}
or Kirsch \cite[Page 305, Theorem 2]{Ki1}).
Next, we find a sequence of operators $H_n$ in ${\cal A}$ such that
$H_n$ converges to the free Laplacian $-\Delta$ in the strong resolvent
sense. 
This implies $\Sigma\supset [0,\infty)$
by \cite[Theorem VIII.24]{ReSi1}, and the converse
inclusion is trivial if the operators in ${\cal A}$ are non-negative.

However, under Assumption \ref{assumptions}, 
finding such a sequence $H_n$ is not an easy task,
because the operator domain $D(H_\omega)$ depends
on the lattice $\Gamma_\omega$ and the flux $\alpha_\gamma$
due to the singularity of the vector potential $\a_\omega$
(the well-known criterion on the strong resolvent convergence 
\cite[Theorem VIII.25]{ReSi1} requires the operators $H_n$
have a common operator core).
Moreover, the unboundedness of the number of the lattice points 
in the basic cell $Q_0$ (such as our example (iii)) 
makes the problem more difficult.

In order to overcome this difficulty,
we directly construct the Weyl sequence for any energy $E \geq 0$, 
i.e., the sequence $\{u_n\}\subset D(H_\omega)$ 
such that $\|(H_\omega -E)u_n\|\to 0$
and $\|u_n\|=1$ ($\|\cdot\|$ denotes the $L^2$-norm), 
by multiplying the factor
\begin{displaymath}
  \Psi(z) = \prod_{\gamma \in \Gamma \cap Q}
 |z - \gamma|^{\alpha_\gamma}
\end{displaymath}
($Q$ is some cube) to the eigenfunction 
$e^{i\sqrt{E} x_1 }$ of $-\Delta$.
Under our small flux assumption,
we can almost surely choose the cube $Q$ so that the magnetic flux on $Q$
is arbitrarily small, and then we can construct 
the desired sequence.
For the detail, see section 2.

Next we shall introduce the integrated density of states (IDS) 
for the operator $H_\omega$.
Let $H_{\omega,N}^k$ be the self-adjoint realization of
the operator ${\cal L}_\omega$ on $L^2(Q_k)$
with the Neumann boundary conditions
$\left(\frac{1}{i}\nabla - \a_\omega\right)u \cdot \n =0$
on $\partial Q_k$
($\n$ is the unit outer normal at the boundary point;
notice also that $\partial Q_k\cap \Gamma_\omega=\emptyset$ almost surely
by (v) of Assumption \ref{assumptions}).
For $E\in \R$, let
$N_{\omega,N}^k (E)$ be the number of the eigenvalues of
$H_{\omega,N}^k$ 
less than or equal to $E$,
and
\begin{displaymath}
 N(E) = \lim_{k\rightarrow \infty}
  \frac{1}{|Q_k|} N_{\omega,N}^k(E),
\end{displaymath}
where $|\cdot|$ denotes the Lebesgue measure.
We can prove the limit $N(E)$ exists almost surely 
and independent of $\omega$,
by Akcoglu-Krengel's superadditive ergodic theorem
(see \cite{C-L, A-K}).

Our second result is the following inequality,
known as the {\it Lifshitz tail} estimate
in the theory of the random Schr\"odinger operators.
\begin{theorem}
 \label{theorem_main}
\begin{enumerate}
 \item 
Suppose Assumption \ref{assumptions} holds.
Then, there exist positive constants $C$ and $E_0$
independent of $\omega$ and $E$, such that
\begin{equation}
\label{main1}
  N(E) \leq e^{-\frac{C}{E}}
\end{equation}
for any $E$ with $0<E<E_0$.

 \item 
Suppose Assumption \ref{assumption2} holds.
Then we have
\begin{equation}
 \label{main2}
\lim_{E \to +0}\frac{\log|\log N(E)|}{\log E}=-1.
\end{equation}
\end{enumerate}
\end{theorem}
\noindent
Notice that the upper bound in (\ref{main2}) is a consequence
from (\ref{main1}). Thus (\ref{main2}) gives a lower bound 
of $N(E)$ in some weak sense.

The Lifshitz tail estimate
is first predicted by I.\ M.\ Lifshitz \cite{L1,L2},
and has been studied in connection with the
mathematical proof of the Anderson localization,
mainly for the Schr\"odinger operators
with random {\it scalar potentials}.
For the reference, see e.g.\ 
\cite{C-L, S1, Ki2, G1, S2, H}.

The mathematical proof of the Lifshitz tail and the Anderson localization
for the Schr\"odinger operators
with random {\it magnetic fields} is comparatively difficult,
mainly because of the following two reasons.
First, the eigenvalues of the operator restricted to
a finite box do \textit{not} depend monotonically 
on the random coupling constants. This fact makes the proof of
the Wegner estimate rather difficult.
Second, especially in the two-dimensional case,
the correlation 
between the values of the magnetic vector potential
at two different points is rather strong,
since 
any vector potential corresponding to a
single-site magnetic field falls off at infinity 
not faster than $O(|x|^{-1})$,
if the total magnetic flux is not zero.

The second difficulty can be solved if 
we assume the single-site magnetic
vector potential is compactly supported.
Under this assumption,
Klopp--Nakamura--Nakano--Nomura \cite{K-N-N-N}
prove the Anderson localization
in the discrete model.
In the continuum model,
Ghribi \cite{G2} proves the internal Lifshitz tail,
and Ghribi--Hislop--Klopp \cite{G-H-K} 
prove the Anderson localization.

In the case the dimension is two 
and the magnetic flux of the single-site 
magnetic field is non-zero,
there is only a few results.
Nakamura \cite{N1, N2} proves the Lifshitz tail 
at the bottom of the spectrum,
both in the discrete and in the continuum model.
Erd\"os--Hasler \cite{E-H1, E-H2, E-H3}
give remarkable results,
in which they succeed to prove 
the Anderson localization
both in the discrete and in the continuum model
(though the form of the magnetic field 
is rather special in the continuum model).
Ueki \cite{U1,U2}
studies the Gaussian random magnetic fields,
and obtain the Wegner estimate in \cite{U3} 
by using the idea of Erd\"os--Hasler.
Hasler--Luckett \cite{H-L} also prove
the Wegner estimate with optimal volume dependence
in the discrete model.

There are also some results for
the random $\delta$ magnetic fields.
Borg--Pul\'e \cite{B-P} prove 
the Lifshitz tail for a smooth approximation of 
a random $\delta$ magnetic field,
but not for the $\delta$ magnetic fields itself.
Borg \cite{B} gives a stochastic representation of
the Laplace transform of IDS
for the Schr\"odinger operator with 
$\delta$ magnetic fields,
in terms of the rotation number of the two-dimensional Brownian motion.
However, there seems to be no results for the Lifshitz tail 
for random $\delta$ magnetic fields at present.

The strategy for the proof of Theorem \ref{theorem_main}
is as follows.
In Nakamura's paper \cite{N1}, 
the crucial inequality in the proof of the Lifshitz tail
is the Avron-Herbst-Simon estimate:
\begin{equation}
 \label{ahs}
 H_\omega \geq \curl \a_\omega.
\end{equation}
If the magnetic field is regular,
we can reduce the problem to
the scalar potential case
by using (\ref{ahs}).
However, in our case the inequality (\ref{ahs})
is no longer useful, since $\curl \a_\omega=0$
almost everywhere.
Instead of (\ref{ahs}), 
we use {\it the Hardy-type inequality}
by Laptev--Weidl \cite{L-W}
(see also (\ref{lw}) below).
Then we can reduce 
the problem to the scalar potential case
as in \cite{N1}.

For the proof of the lower bound,
we follow the standard strategy
given in \cite[Theorem VI.2.7]{C-L}.
We give an estimate for the probability 
of the event `the lowest eigenvalue 
of the Dirichlet realization $H_{\omega, D}^k$ of the 
operator $H_\omega$ on $Q_k$ is less than $\epsilon$',
by constructing an approximation of the ground state
explicitly.
Here we use the estimates
obtained in the proof of Theorem \ref{theorem_spectrum}.

The rest of the paper is organized as follows.
In section 2, we shall prove Theorem \ref{theorem_spectrum},
and the lower bound in Theorem \ref{theorem_main}.
In section 3, we shall introduce the Hardy-type inequality 
by Laptev--Weidl \cite{L-W}, 
and give some key inequality in the proof of
Theorem \ref{theorem_main}.
In section 4, we shall prove Theorem \ref{theorem_main}.

%
%
%
%
\section{Spectrum}
In this section, we shall give a proof 
of Theorem \ref{theorem_spectrum},
using the method of approximating eigenfunctions.
First we review a lemma about the
gauge transform for $\delta$ magnetic fields.
\begin{lemma}
\label{gauge}
 Let $U$ be a simply connected open set in $\R^2$
 and $\Gamma$ be a discrete subset of $U$
without accumulation points in $U$.
 Let $\a, \tilde{\a}\in 
C^\infty(U\setminus \Gamma;\R^2)\cap L^1_{\rm loc}(U;\R^2)$.
Assume 
\begin{displaymath}
 \curl \a = \curl \tilde{\a} + 
\sum_{\gamma\in \Gamma}2\pi n_\gamma \delta_\gamma
\end{displaymath}
holds in ${\cal D}'(U)$,
where $n_\gamma\in \mathbb{Z}$
and $\delta_\gamma$ is the Dirac measure 
supported on the point $\gamma$.
Then, there exists $\Phi\in C^\infty(U\setminus \Gamma)$ such that
$|\Phi(z)|=1$ for any $z\in U\setminus \Gamma$ and
\begin{displaymath}
 \left(\frac{1}{i}\nabla - \a \right)(\Phi u)
= \Phi \left(\frac{1}{i}\nabla - \tilde{\a} \right) u
\end{displaymath}
for any $u\in C^\infty(U\setminus\Gamma)$.
\end{lemma}
For a proof, see \cite[Theorem 3.1]{G-S}.
By Lemma \ref{gauge},
we can arbitrarily choose an appropriate vector potential
for given $\delta$-magnetic fields.

\begin{proof}[Proof of Theorem \ref{theorem_spectrum}]
Since the operator $H_\omega$ is non-negative, 
we have $\sigma(H_\omega)\subset [0,\infty)$.
In order to prove $\sigma(H_\omega)\supset [0,\infty)$,
we shall prove $\xi^2\in \sigma(H_\omega)$  almost surely
for any $\xi\in \mathbb{Q}$.
Then it suffices to show that we can 
almost surely
find a function $u$ satisfying
\begin{equation}
\label{2-1}
 u\in D(H_\omega),\quad 
\|u\|=1,\quad
\|(H_\omega - \xi^2)u\|<\epsilon
\end{equation}
for any positive rational number $\epsilon$,
where the norm without suffix denotes the $L^2$-norm.

Take sufficiently large positive integers $k$ and $l$,
which will be determined later.
By (vi) of Assumptions \ref{assumptions},
we can almost surely find 
$n\in \mathbb{Z}^2$
such that the square
$Q=Q_k + n$ satisfies
\begin{eqnarray}
\label{2-2}
&& \Phi_\omega(m + Q_0)<\frac{1}{l}\quad
\mbox{for any }m \in Q \cap\mathbb{Z}^2,\\
&&B_\gamma(c\sqrt{\alpha_\gamma})\cap
B_{\gamma'}(c\sqrt{\alpha_{\gamma'}})=\emptyset,\quad
B_\gamma(c\sqrt{\alpha_\gamma})\cap \partial (m+Q_0) =\emptyset\nonumber\\
&&
\quad \mbox{for any }
m\in Q \cap\mathbb{Z}^2,\ 
\gamma,\gamma'\in (m+Q_0)\cap \Gamma_\omega
\mbox{ with }\gamma\not=\gamma',
\label{2-3}
\end{eqnarray}
or
\begin{equation}
\label{2-4}
 \Phi_\omega(m + Q_0)\leq
\sum_{\gamma\in (m+Q_0)\cap \Gamma_\omega}\sqrt{\alpha_\gamma} 
<\frac{1}{l}\quad
\mbox{for any }m\in Q \cap\mathbb{Z}^2.
\end{equation}
We omit the random parameter $\omega$ in the rest of the proof,
because we do not use the probabilistic argument hereafter.

We shall construct the function $u$
supported in the square $Q$.
By Lemma \ref{gauge} and (\ref{1-2}), we may assume
\begin{displaymath}
 \a (z)  =(\Im \psi(z), \Re \psi(z)),\quad
\psi(z) = \sum_{\gamma \in \Gamma \cap Q}
 \frac{\alpha_\gamma}{z-\gamma}
\end{displaymath}
in $Q$, where we again identify 
$x=(x_1,x_2)$ with $z=x_1+ix_2$
and regard $\gamma$ as a complex number.
Put
\begin{displaymath}
  \Psi(z) = \prod_{\gamma \in \Gamma \cap Q}
 |z - \gamma|^{\alpha_\gamma}.
\end{displaymath}
If $\Gamma\cap Q=\emptyset$, we put 
$\psi(z)=0$ and $\Psi(z)=1$.
Then we have
\begin{equation}
 2 \partial_{\bar{z}} \Psi(z) = \overline{\psi(z)} \Psi(z),\quad
 - 2 \partial_{z} \Psi(z)^{-1} 
= \psi(z) \Psi(z)^{-1},
\label{2-4.1}
\end{equation}
where $\partial_z = (\partial_1 - i\partial_2)/2$,
$\partial_{\bar{z}} = (\partial_1 + i\partial_2)/2$.
Thus we have by (\ref{2-4.1})
\begin{eqnarray}
 {\cal L} 
&=& \left(\frac{1}{i}\nabla - \a\right)^2\nonumber\\
&=&    \left(2\partial_{\bar{z}} + \overline{\psi}\right)
\left(-2\partial_{z} + \psi\right)\nonumber\\
&=& \Psi^{-1} (2\partial_{\bar{z}}) 
   \Psi^2 (-2\partial_{z})\Psi^{-1},
\label{2-5}
\end{eqnarray}
as an operator acting on the functions on $Q\setminus \Gamma$.

Let $\chi_k\in C_0^\infty(Q)$ 
satisfying the following conditions:
\begin{eqnarray*}
 &&0 \leq \chi_k(z) \leq 1,\quad
 \chi_k (z) = 
\begin{cases}
 1 & (z\in n + Q_{k-1}), \cr
 0 & (z\in n 
+ (Q_k \setminus Q_{k-1/2})),
\end{cases}
\\ 
 &&\|\nabla\chi_k\|_\infty + \|\Delta \chi_k\|_\infty \leq C_0,
\end{eqnarray*}
where $C_0$ is a constant independent of $k,n$.
Put
\begin{displaymath}
 v_k = \chi_k \Psi e^{ix_1 \xi}, \quad
u_k=\frac{v_k}{\|v_k\|}.
\end{displaymath}
We are going to show $u=u_k$ satisfies (\ref{2-1}),
if we take $k$ and $l$ sufficiently large.

Let us estimate $\|v_k\|$ from below.
Take $m \in (n+Q_{k-1}) \cap \mathbb{Z}^2$,
and put
\begin{eqnarray*}
 \Gamma_1 &=& \Gamma\cap (m + Q_1),\\
 \Gamma_2 &=& \Gamma\cap (Q\setminus (m + Q_1)).
\end{eqnarray*}
Let 
$z\in m+Q_0$.
Using the inequality
$e^t \geq 1+t$ $(t\in \mathbb{R})$,
we have
\begin{eqnarray}
 \prod_{\gamma\in \Gamma_1}|z-\gamma|^{2\alpha_\gamma}
&=& \exp
\left(\sum_{\gamma\in \Gamma_1}2\alpha_\gamma \log |z-\gamma|\right)
\nonumber\\
&\geq&
1 + \sum_{\gamma\in \Gamma_1}2\alpha_\gamma \log |z-\gamma|,
\label{2-6}\\
 \prod_{\gamma\in \Gamma_2}|z-\gamma|^{2\alpha_\gamma}
&\geq & 1.
\label{2-7}
\end{eqnarray}
Notice that
\begin{equation}
 \int_S \log|z| dx\geq \int_{|z|\leq 1}\log |z|dx
=2\pi \int_0^1 r\log r dr = -\frac{\pi}{2}
\label{2-8}
\end{equation}
for any bounded Borel set $S$.
By 
(\ref{2-6}), (\ref{2-7}), (\ref{2-8}),
and (\ref{2-2}) or (\ref{2-4}),
we have
\begin{eqnarray}
 \int_{m+Q_0}
|v_k|^2 dx\nonumber
&\geq&
1+
\sum_{\gamma \in \Gamma_1}
\int_{m+Q_0}
2\alpha_\gamma \log |z-\gamma|dx\nonumber\\
&\geq&
1
-\pi
\sum_{\gamma \in \Gamma_1}\alpha_\gamma\nonumber\\
&\geq& 1 - \frac{9\pi}{l}
\label{2-9}
\end{eqnarray}
for any $m \in n + Q_{k-1}$.
Take $l_0$ so large that $1-9\pi/l_0>0$.
Then, summing up the both sides of (\ref{2-9})
with respect to $m$,
we find a positive constant $C_1$
independent of $k,l$ such that
\begin{equation}
 \label{2-10}
 \|v_k\| \geq C_1(2k+1)
\end{equation}
for any $l \geq l_0$.

Next we shall estimate $\|(H - \xi^2)v_k\|$ from above.
Since
\begin{displaymath}
 2 \partial_z e^{ix_1\xi} = 2 \partial_{\bar{z}} e^{ix_1\xi} = 
\partial_1 e^{ix_1\xi}=i\xi e^{ix_1\xi},
\end{displaymath}
we have by (\ref{2-4.1}) and (\ref{2-5})
\begin{eqnarray}
 {\cal L} v_k
&=&
\Psi^{-1} (2\partial_{\bar{z}}) 
   \Psi^2 (-2\partial_{z})
( \chi_k e^{ix_1\xi})\nonumber\\
&=&
\Psi^{-1} (2\partial_{\bar{z}}) 
   \Psi^2 \left(-2\partial_{z} \chi_k
                   - i\xi \chi_k 
                   \right)e^{ix_1\xi}\nonumber\\
&=& -2\Psi \overline{\psi}
    \left( 2\partial_{z}\chi_k + i\xi \chi_k \right) e^{ix_1\xi}\nonumber\\
&& - \Psi\left(
 \Delta \chi_k
 + 2i\xi \partial_1 \chi_k
\right)e^{ix_1\xi} + \xi^2 v_k.
\label{2-11}
\end{eqnarray}
(\ref{2-11}) implies the singularity of ${\cal L}v_k$ near
$\gamma\in Q\cap \Gamma$ is at most
$O(|z-\gamma|^{-1+\alpha_\gamma})$,
so ${\cal L}v_k\in L^2(\mathbb{R}^2)$
and $v_k \in D(H)$ by (\ref{bdry}).
By (\ref{2-11}), we have
\begin{eqnarray}
\label{2-12}
 \|(H - \xi^2)v_k\|
\leq  C_2 \left(
   \|\Psi \psi\|_{L^2(Q)}
+  \|\Psi\|_{L^2(n+ (Q_k \setminus Q_{k-1}))}
\right),
\end{eqnarray}
where $C_2 =2(C_0+1)(|\xi|+1)$.

In the sequel, we denote
 $d_k =\diam Q = \sqrt{2}(2k+1)$.
Then we have
\begin{equation}
 \|\Psi\|_{L^2(n+(Q_k \setminus Q_{k-1}))}
\leq \|\Psi\|_\infty 
\left|Q_k \setminus Q_{k-1}\right|^{\frac{1}{2}}
\leq d_k^{\Phi(Q)}\sqrt{8k}.
\label{2-21}
\end{equation}
We shall estimate $\|\Psi \psi\|_{L^2(Q)}$
under the conditions (\ref{2-2}) and (\ref{2-3}),
and under the condition (\ref{2-4}), separately.
In both cases, we have
\begin{equation}
\label{2-12.5}
 \Phi(Q)<\frac{(2k+1)^2}{l},
\end{equation}
and
\begin{equation}
 |\Psi\psi(z)|
\leq \sum_{\gamma\in \Gamma\cap Q}
\alpha_\gamma |z-\gamma|^{\alpha_\gamma-1}
\prod_{\gamma'\not=\gamma}
|z-\gamma'|^{\alpha_{\gamma'}}
\label{2-13}
\end{equation}
for $z\in Q$,
where $\prod_{\gamma'\not=\gamma}$ stands for
$\prod_{\gamma'\in \Gamma\cap Q,\,\gamma'\not=\gamma}$
as in the sequel.

Assume (\ref{2-2}) and (\ref{2-3}) hold.
By (\ref{2-13}), we have
\begin{eqnarray}
&&
\|\Psi \psi\|_{L^2(\bigcup_\gamma B_{\gamma}(c\sqrt{\alpha_\gamma}))}^2\nonumber\\
 &\leq&
\sum_{\gamma\in \Gamma\cap Q}
2\alpha_\gamma^2
\int_{B_{\gamma}(c\sqrt{\alpha_\gamma})}
|z-\gamma|^{2\alpha_\gamma-2}dx\cdot
d_k^{2\Phi(Q)}\nonumber\\
&&+
\sum_{\gamma\in \Gamma\cap Q}
2
\int_{B_{\gamma}(c\sqrt{\alpha_\gamma})}
\left(
\sum_{\mu\not=\gamma}\alpha_{\mu}|z-\mu|^{\alpha_{\mu}-1}\right)^2dx
\cdot d_k^{2\Phi(Q)}.
\label{2-14}
\end{eqnarray}
The first sum in (\ref{2-14}) is bounded by
\begin{eqnarray}
&& 
\sum_{\gamma\in \Gamma\cap Q}
2\alpha_\gamma^2\cdot 2\pi\int_0^{c\sqrt{\alpha_\gamma}}
r^{2\alpha_\gamma-2}\cdot rdr
\cdot d_k^{2\Phi(Q)}\nonumber\\
&=& 
\sum_{\gamma\in \Gamma\cap Q}
2\pi \alpha_\gamma (c\sqrt{\alpha_\gamma})^{2\alpha_\gamma}\cdot
d_k^{2\Phi(Q)}\nonumber\\
&\leq&
2\pi \Phi(Q) d_k^{2\Phi(Q)}
\label{2-15}
\end{eqnarray}
where we used $0<c\leq 1$ and $0\leq \alpha_\gamma<1$
in the last inequality.

For the second sum in (\ref{2-14}),
we use the inequality
\begin{displaymath}
\frac{1}{2} \leq\frac{|z-w|}{|z-\mu|}\leq \frac{3}{2}
\end{displaymath}
for  $z\not\in B_\mu(c\sqrt{\alpha_\mu})$
and  $w \in B_\mu(c\sqrt{\alpha_\mu}/2)$.
Then we see that
there exists a positive constant $C_3$ 
independent of $\mu$ and $\alpha_\mu$ with $0\leq \alpha_\mu <1$ such that
\begin{equation}
\label{2-16}
 \alpha_\mu|z-\mu|^{\alpha_\mu-1}
\leq C_3\int_{B_\mu(c\sqrt{\alpha_\mu}/2)}
 |z-w|^{\alpha_\mu-1}du
\end{equation}
for $z\not\in B_\mu(c\sqrt{\alpha_\mu})$,
where we write $w=u_1+i u_2$ and $du=du_1du_2$.
Summing up (\ref{2-16}) with respect to $\mu\not=\gamma$,
we have
\begin{eqnarray}
\sum_{\mu\not=\gamma}
\alpha_\mu|z-\mu|^{\alpha_\mu-1}
\leq 
C_3
\int_{\bigcup_{\mu\not=\gamma}B_\mu(c\sqrt{\alpha_\mu}/2)}
(\min(|z-w|,1)^{-1}du
\label{2-17}
\end{eqnarray}
for $z\in B_\gamma(c\sqrt{\alpha_\gamma})$,
where we used the disjointness assumption (\ref{2-3}).
The area of the domain of integration in (\ref{2-17})
is bounded by
\begin{displaymath}
\frac{\pi c^2}{4}\sum_{\mu\not=\gamma}\alpha_\mu\leq\frac{\pi c^2}{4}\Phi(Q),
\end{displaymath}
and the right-hand side equals the area 
of the disc of radius $c\sqrt{\Phi(Q)}/2$.
Since the integrand of (\ref{2-17}) is monotone non-increasing
with respect to $|z-w|$,
we have by (\ref{2-17})
\begin{eqnarray}
\sum_{\mu\not=\gamma}
\alpha_\mu|z-\mu|^{\alpha_\mu-1}
&\leq& C_3
 \int_{|z-w|\leq c\sqrt{\Phi(Q)}/2}|z-w|^{-1}du\nonumber\\
&=& \pi c C_3\sqrt{\Phi(Q)},
\label{2-17.1}
\end{eqnarray}
provided that $l$ is sufficiently large so that 
\begin{equation}
\label{2-17.2}
\frac{c\sqrt{\Phi(Q)}}{2}\leq \frac{c(2k+1)}{2\sqrt{l}}\leq 1,
\end{equation}
where we used (\ref{2-12.5}).
By (\ref{2-17.1}),
the second sum in (\ref{2-14}) 
is bounded by
\begin{equation}
\label{2-18}
\sum_{\gamma\in \Gamma\cap Q} 
2 \cdot \pi c^2 \alpha_\gamma 
\cdot (\pi c C_3)^2\Phi(Q) \cdot
d_k^{2\Phi(Q)}
=
C_4 \Phi(Q)^2 d_k^{2\Phi(Q)}
\end{equation}
for sufficiently large $l$ satisfying (\ref{2-17.2}),
where $C_4=2\pi^3 c^4 C_3^2$.

For $z \in Q\setminus \bigcup_{\gamma} B_\gamma(c \sqrt{\alpha_\gamma})$,
we have estimate
(\ref{2-16})
for every $\mu \in Q \cap \Gamma$.
Repeating the above argument, we see that
\begin{displaymath}
 |\Psi\psi(z)|\leq \pi c C_3\sqrt{\Phi(Q)}\cdot d_k^{\Phi(Q)}
\end{displaymath}
for sufficiently large $l$ satisfying (\ref{2-17.2}), and 
\begin{equation}
\label{2-19}
 \|\Psi\psi\|^2_{L^2(Q\setminus \bigcup_{\gamma} B_\gamma(c \sqrt{\alpha_\gamma}))}
\leq (\pi c C_3)^2\Phi(Q) d_k^{2\Phi(Q)} \cdot (2k+1)^2,
\end{equation}
since the area of $Q$ is $(2k+1)^2$.
By (\ref{2-12.5}), (\ref{2-14}), (\ref{2-15}), 
(\ref{2-18}) and (\ref{2-19}),
there exists a positive constant $C_5$ independent of $k$ and $l$
such that
\begin{equation}
\label{2-20}
 \|\Psi\psi\|_{L^2(Q)} \leq 
C_5\frac{(2k+1)^2}{\sqrt{l}}d_k^{\Phi(Q)},
\end{equation}
for sufficiently large $l$ satisfying (\ref{2-17.2}).

By (\ref{2-10}),
(\ref{2-12}), (\ref{2-21}) and (\ref{2-20}),
we have
\begin{equation}
\label{2-22}
 \|(H -\xi^2)u_k\| \leq
C_1^{-1} C_2 
d_k^{(2k+1)^2/l}
\left(
\frac{C_5(2k+1)}{\sqrt{l}}
+ 
\frac{\sqrt{8k}}{2k+1}
\right)
\end{equation}
for sufficiently large $l$ satisfying $l \geq l_0$
and (\ref{2-17.2}).
For $\epsilon>0$,
take $k$ so large that
$C_1^{-1}C_2\sqrt{8k}/(2k+1) < \epsilon/2$,
and then take $l$ so large that
the right hand side of $(\ref{2-22})$ is less than $\epsilon$.
Then $u_k$ satisfies (\ref{2-1}).

Next we assume (\ref{2-4}) holds. 
By the Minkowski inequality, we have
\begin{eqnarray}
 && \|\Psi \psi\|_{L^2(Q)}\nonumber\\
 &\leq&
 \sum_{\gamma \in Q\cap \Gamma}
\alpha_\gamma
 \left(
 \int_Q  \left| z-\gamma \right|^{2\alpha_\gamma-2}
 \prod_{\gamma'\neq \gamma} \left| z-\gamma'\right|^{2\alpha_{\gamma'}}
dx 
\right)^{\frac{1}{2}}\nonumber\\
&\leq&
 \sum_{\gamma \in Q\cap \Gamma}
\alpha_\gamma
d_k^{\Phi(Q)-\alpha_\gamma}
 \left(
2\pi \int_0^{d_k}r^{2\alpha_\gamma-2}\cdot rdr
\right)^{\frac{1}{2}}\nonumber\\
&=&
 \sum_{\gamma \in Q\cap \Gamma}
\sqrt{\alpha_\gamma \pi}
d_k^{\Phi(Q)}\nonumber\\
&\leq&
\frac{\sqrt{\pi}}{l}
(2k+1)^2
d_k^{\Phi(Q)}.\label{2-23}
\end{eqnarray}
By (\ref{2-10}), 
(\ref{2-12}), (\ref{2-21}), (\ref{2-12.5}) and (\ref{2-23}),
we have
\begin{displaymath}
 \|(H -\xi^2)u_k\| \leq
C_1^{-1} C_2 
d_k^{(2k+1)^2/l}
\left(
\frac{\sqrt{\pi}(2k+1)}{l}
+ 
\frac{\sqrt{8k}}{2k+1}
\right)
\end{displaymath}
for sufficiently large $l$,
and obtain the conclusion similarly.
\end{proof}

Using the estimates obtained in the above proof,
we can also prove the lower bound in (ii) of Theorem \ref{theorem_main}.
We use the same notation introduced above.
\begin{proof}[Proof of the lower bound in (ii) of Theorem \ref{theorem_main}]
Let $H_{\omega,D}^k$ be the operator $H_\omega$ restricted on
$L^2(Q_k)$ with the Dirichlet boundary conditions,
and put $N_{\omega, D}^k(E)$ be the number of the eigenvalues
of $H_{\omega,D}^k$ less than or equal to $E$, counted with multiplicity.
It is known that
\begin{displaymath}
 N(E)=\sup_{k \geq 1}\frac{1}{|Q_k|}N_{\omega,D}^k(E)
\end{displaymath}
holds almost surely (see \cite[VI.1.3]{C-L}).
Let $E_1(H)$ denotes the smallest eigenvalue
of a self-adjoint operator $H$.
Then we have
\begin{eqnarray}
 N(E)
&\geq& \frac{1}{|Q_k|}\mathbf{E}\left[N_{\omega,D}^k(E)\right]
\nonumber\\
&\geq& \frac{1}{|Q_k|}\mathbf{P}\left[E_1(H_{\omega,D}^k)\leq E\right]
\label{lower01}
\end{eqnarray}
for any $k \geq 1$.

We assume (\ref{assumption2_1}) holds for sufficiently small $\epsilon>0$,
and estimate the right hand side of (\ref{lower01}) from below.
Let $\epsilon$ be a small positive number.
Take positive integers $k$ and $l$ so that
\begin{equation}
 \label{lower05}
\frac{\epsilon}{8}< \left(\frac{\pi}{2k+1}\right)^2<\frac{\epsilon}{4},
\quad
\frac{\epsilon^3}{2}< \frac{1}{l} < \epsilon^3.
\end{equation}
Then (\ref{2-17.2}) is satisfied for sufficiently small $\epsilon$.

Let us suppose
(\ref{2-2}) and (\ref{2-3}) hold with $Q=Q_k$,
and construct an approximation of the ground state of 
$H_{\omega,D}^k$. Let $-\Delta_D^k$ be the Dirichlet Laplacian
on $Q_k$. The ground state of $-\Delta_D^k$ is
\begin{displaymath}
 f_k(x) 
= 
\cos\left(\frac{\pi}{2k+1}x_1\right)
\cos\left(\frac{\pi}{2k+1}x_2\right),
\end{displaymath}
with the ground state energy $2 \cdot \left(\pi/(2k+1)\right)^2$.
Put
\begin{displaymath}
 w_k = \Psi f_k,\quad u_k = \frac{w_k}{\|w_k\|}.
\end{displaymath}
Then, similar to (\ref{2-10}), we can prove 
\begin{equation}
 \label{lower02}
 \|w_k\| \geq C_1'(2k+1)
\end{equation}
for any $l \geq l_0$ and $k \geq 1$,
 by using the fact $|f_k|\geq 1/4$ on $Q_{k/2}$.
By (\ref{2-4.1}) and (\ref{2-5}), we have
\begin{displaymath}
 H_{\omega,D}^k w_k 
=
2\left(\frac{\pi}{2k+1}\right)^2w_k
+
2 \Psi \overline{\psi}\cdot(-2\partial_z f_k).
\end{displaymath}
Since $|2\partial_z f_k|\leq \pi/(2k+1)$,
we have by (\ref{lower02})
\begin{eqnarray}
 (u_k,H_{\omega,D}^k u_k)
&\leq& 
2\left(\frac{\pi}{2k+1}\right)^2
+ \frac{2\pi {C_1'}^{-1}}{(2k+1)^2} 
\|\Psi \psi\|,
\label{lower04}
\end{eqnarray}
where $(u,v)=\int \bar{u}v dx$ 
denotes the $L^2$-inner product.
Then (\ref{lower04}), (\ref{2-12.5}), (\ref{2-20}) and (\ref{lower05}) imply
\begin{displaymath}
 d_k =\sqrt{2}(2k+1)<4\pi\epsilon^{-1/2},\quad
\Phi(Q)<\frac{(2k+1)^2}{l}<8\pi^2\epsilon^2,
\end{displaymath}
\begin{eqnarray*}
 (u_k,H_{\omega,D}^k u_k)
&\leq& 
2\left(\frac{\pi}{2k+1}\right)^2
+ \frac{2\pi {C_1'}^{-1}C_5}{\sqrt{l}} 
d_k^{\Phi(Q)}\\
&\leq &
\epsilon
\left(
\frac{1}{2}+ 2\pi {C_1'}^{-1}C_5\epsilon^{1/2}
\left(4\pi\epsilon^{-1/2}\right)^{8\pi^2\epsilon^2}
\right).
\end{eqnarray*}
Since the expression in the big parenthesis tends to $1/2$
as $\epsilon$ tends to $0$, we have by the min-max principle
\begin{equation}
 \label{lower06}
E_1(H_{\omega,D}^k)<\epsilon
\end{equation}
for sufficiently small $\epsilon$.
Thus the events
(\ref{2-2}) and (\ref{2-3}) with $k$ and $l$ satisfying
(\ref{lower05}) imply the inequality (\ref{lower06}).
Then the independentness assumption implies
\begin{eqnarray}
 \mathbf{P}\{E_1(H_{\omega,D}^k \leq \epsilon)\}
&\geq & 
\mathbf{P}\{\mbox{(\ref{2-2}) and (\ref{2-3}) holds}\}\nonumber\\
&\geq& (c_1 l^{-\delta_1})^{|Q_k|}\nonumber\\
&\geq& (c_1 (\epsilon^3/2)^{\delta_1})^{8\pi^2 \epsilon^{-1}}
\label{lower07}
\end{eqnarray}
for sufficiently small $\epsilon>0$.
Then (\ref{lower01}), (\ref{lower05}) and (\ref{lower07}) imply
\begin{displaymath}
\liminf_{E\to +0}  
\frac{\log | \log N(E)|}{\log E}
\geq -1,
\end{displaymath}
which is the desired conclusion.
We can give the same conclusion in the case 
(\ref{assumption2_2}) holds for sufficiently small $\epsilon>0$,
by using (\ref{2-23}) instead of (\ref{2-20}).
\end{proof}

%
%
%
%
\section{Hardy-type inequality}
For $d\geq 3$,
{\it the Hardy inequality} says
\begin{equation}
\label{hardy}
 \int_{\mathbb{R}^d}|\nabla u|^2 dx
 \geq \left(\frac{d-2}{2}\right)^2 
\int_{\mathbb{R}^d} \frac{|u|^2}{|x|^2} dx
\end{equation}
for any $u\in H^1(\mathbb{R}^d)$.
The inequality (\ref{hardy}) fails when $d=2$,
however, Laptev--Weidl \cite{L-W}
prove that a similar inequality holds
if there exists a $\delta$ magnetic field
at the origin.

\begin{lemma}[Laptev-Weidl]
\label{laptev_weidl} 
Let $\alpha\in\mathbb{R}$ and put 
$\displaystyle 
\a_\alpha= \left(\Im \frac{\alpha}{z}, \Re \frac{\alpha}{z}\right)$,
where $z=x_1+ix_2$
($a_\alpha$ satisfies $\curl \a_\alpha = 2\pi \alpha\delta$).
Then, we have
\begin{equation}
 \label{lw}
 \int_{|x|\leq R} \left| \left(
 \frac{1}{i}\nabla - \a_\alpha\right) u\right|^2 dx
 \geq
 \rho(\alpha)
 \int_{|x|\leq R}
 \frac{|u|^2}{|x|^2} dx
\end{equation}
for any $R>0$ and any $u\in C_0^\infty(\mathbb{R}^2\setminus \{0\})$,
where $\rho(\alpha)=\min_{n\in \mathbb{Z}}|n-\alpha|^2$.
\end{lemma}

Let us return to our model.
We define 
a random scalar potential $V_\omega(z)$ as follows.
For $m \in \mathbb{Z}^2$,
put $\Gamma_{\omega,m}=(m+Q_0)\cap \Gamma_\omega$
and
\begin{displaymath}
 \delta_m =
\min_{\gamma\in \Gamma_{\omega,m}}
\left(
\dist(\gamma, ( m+\partial Q_0)\cup (\Gamma_{\omega,m}\setminus \{\gamma\})
\right)/2.
\end{displaymath}
For $z \in m + Q_0$, put
\begin{displaymath}
 V_\omega(z) = 
\begin{cases}
\displaystyle
\min\left(
\frac{1}{\delta_m^2}
 \rho(\alpha_\gamma(\omega)),
1
\right)&
 \mbox{(if $z \in B_{\gamma}(\delta_m)$ for some 
$\gamma\in \Gamma_{\omega,m}$)},\cr
 0 & (\mbox{otherwise}).
\end{cases}
\end{displaymath}
 \begin{lemma}
\label{laptev_weidl2} 
Let $H_{\omega,N}^k$ be the Neumann realization 
of ${\cal L}_\omega$ on $Q_k$,
and $\Delta_N^k$ the Neumann Laplacian on $Q_k$.
Then,
\begin{displaymath}
E_1(H_{\omega,N}^k) \geq  E_1\left(
\frac{1}{2}\left(-\Delta_N^k + V_\omega\right)
\right).
\end{displaymath}
\end{lemma}
\begin{proof}
By using Lemma \ref{gauge} and Lemma \ref{laptev_weidl},
we have
\begin{equation}
 \label{lw2}
 \int_{m + Q_0} \left| \left(
 \frac{1}{i}\nabla - \a_\omega\right) u(z)\right|^2 dx
 \geq
 \int_{m + Q_0}
 V_\omega(z) |u(z)|^2 dx dy
\end{equation}
for any $u\in C_0^\infty(\mathbb{R}^2\setminus \Gamma_\omega)$
and any $m\in \mathbb{Z}^2$.
Notice that\footnote{
The equality
(\ref{absolute}) makes sense
for $u\in H^{1,1}_{\rm loc}(\mathbb{R}^2)$; see e.g.\
\cite[appendix]{L-S}.
}
\begin{equation}
\label{absolute}
\nabla |u| = \Re (\sgn \bar{u} \nabla u)
 = \Re (\sgn \bar{u} (\nabla - i\a_\omega )u)\qquad
\mbox{a.e.}
\end{equation}
holds for $u\in C_0^\infty(\mathbb{R}^2\setminus \Gamma_\omega)$,
where
$\sgn z = z/|z|$ for $z\neq 0$ and $\sgn 0 = 0$.
Taking the absolute value of the both sides, we have
\begin{equation}
\label{lw3}
 \left|\left(\frac{1}{i}\nabla - \a_\omega\right) u\right|^2
 \geq \Bigl|\nabla |u|\Bigr|^2\qquad
\mbox{a.e.}
\end{equation}
By (\ref{lw2}) and (\ref{lw3}), 
we have
\begin{equation}
 \label{lw4}
 \int_{Q_k}\left|
 \left(\frac{1}{i}\nabla - \a_\omega\right)u
\right|^2 dx
\geq \frac{1}{2}
 \int_{Q_k}
 \left(\Bigl|\nabla |u|\Bigr|^2
+ V_\omega\left|u\right|^2
\right) dx
\end{equation}
for any $u\in C_0^\infty (\mathbb{R}^2\setminus \Gamma_\omega)$.
Then the conclusion follows immediately from (\ref{lw4}) and
the min-max principle.
\end{proof}

\section{Proof of Theorem \ref{theorem_main}}
By virtue of Lemma \ref{laptev_weidl2},
we can prove Theorem \ref{theorem_main}
as in the scalar potential case
(this idea is taken from Nakamura's paper \cite{N1}).
The proof below is based on Stollmann's book \cite{S1}.

\begin{lemma}
\label{rough_estimate}
 There exists a positive constant $C_6$ 
independent of $\Gamma_\omega$, $\alpha(\omega)$, E,
and $k$ such that
\begin{displaymath}
 N_{\omega,N}^k(E)
 \leq C_6 |Q_k|
\end{displaymath}
for any $E \leq 1$ and any non-negative integer $k$.
\end{lemma}
\begin{proof}
We use the diamagnetic inequality
\begin{equation}
 \label{diamagnetic}
 |(H_{\omega,N}^k + \lambda)^{-1} u |(x) 
\leq (-\Delta_N^k +\lambda)^{-1}|u|(x)\qquad\mbox{a.e.}
\end{equation}
for any $\lambda>0$ and $u\in L^2(Q_k)$,
where $\Delta_N^k$ is the Neumann Laplacian 
on $Q_k$
\footnote{
The diamagnetic inequality on the whole plane
is proved for $\mathbf{a}\in L^2_{\rm loc}(\mathbb{R}^2;\mathbb{R}^2)$
by Leinfelder--Simader \cite[Lemma 6]{L-S},
and for the Aharonov-Bohm field
by Melgaard--Ouhabaz--Rozenblum 
\cite[Theorem 1.1]{M-O-R}.
The diamagnetic inequality (\ref{diamagnetic})
for the Schr\"odinger operators with 
the Neumann boundary conditions
can be proved similarly,
by slightly changing the proof 
as the functions in the proof belong to the appropriate 
form/operator domain
(a similar argument is seen in 
Doi--Iwatsuka--Mine \cite[Proposition 3.2]{D-I-M}).
}.

Taking $u$ as an approximation of 
the Dirac measure in (\ref{diamagnetic}), 
we have
\begin{displaymath}
\left|(H_{\omega,N}^k + \lambda)^{-1}(x,y)\right| 
 \leq 
 (-\Delta_N^k + \lambda)^{-1}(x,y)\qquad\mbox{a.e.},
\end{displaymath}
where $T(x,y)$ denotes the integral kernel of an integral operator $T$.
This implies
\begin{equation}
\label{rough1}
 \|(H_{\omega,N}^k + \lambda)^{-1}\|_{{\cal I}_2}
 \leq
 \|(-\Delta_N^k + \lambda)^{-1}\|_{{\cal I}_2},
\end{equation}
where $\|\cdot\|_{{\cal I}_2}$ denotes
the Hilbert-Schmidt norm.
Let $l=2k+1$ be the length of the edge of $Q_k$.
By (\ref{rough1}), we have for any $E\leq 1$
\begin{eqnarray*}
 N_{\omega,N}^k(E) 
 &\leq& \tr \chi_{(-\infty, 1]}(H_{\omega,N}^k)\\
 &\leq & \tr 4(H_{\omega,N}^k + 1)^{-2}\\
 &=& 4 \|(H_{\omega,N}^k + 1)^{-1}\|_{{\cal I}_2}^2\\
 &\leq& 4 \|(-\Delta_N^k + 1)^{-1}\|_{{\cal I}_2}^2\\
 &=&  4|Q_k| \sum_{m=0}^\infty\sum_{n=0}^\infty
        F\left(\frac{m}{l}, \frac{n}{l}\right)\frac{1}{l^2},
\end{eqnarray*}
where $F(x,y)= 1/(\pi^2(x^2+y^2)+1)^2$. The double sum in the
last expression converges to
\begin{displaymath}
 \int_0^\infty \int_0^\infty F(x,y)dx <\infty
\end{displaymath}
as $k\rightarrow \infty$, 
so it is bounded with respect to $k$.
Thus we have the conclusion.
\end{proof}

It is known that
\begin{displaymath}
 N(E) = \inf_{k \geq 1} \frac{1}{|Q_k|}
        \mathbf{E}\left[N_{\omega,N}^k(E)\right]
\end{displaymath}
holds almost surely (see \cite[VI.1.3]{C-L}).
Let $E_1(H_{\omega,N}^k)$ be the smallest eigenvalue of $H_{\omega,N}^k$,
and $\chi(\omega)$ the characteristic function
of the event `$E_1(H_{\omega,N}^k) \leq E$'.
Then we have for every $k\geq 1$ and $E\leq 1$
\begin{eqnarray}
  N(E) &\leq& \frac{1}{|Q_k|} \mathbf{E}[N_{\omega,N}^k(E)]\nonumber\\
       &=& \frac{1}{|Q_k|} 
           \mathbf{E}[N_{\omega,N}^k(E)\chi(\omega)]\nonumber\\
       &\leq& C_6 \mathbf{P}\{E_1(H_{\omega,N}^k) \leq E\}\nonumber\\
       &\leq& C_6 \mathbf{P}\left\{E_1\left(
\frac{1}{2}(-\Delta_N^k + V_\omega)
\right) \leq E\right\},
\label{prmain1.5}
\end{eqnarray}
where we used 
Lemma \ref{rough_estimate}
in the second inequality,
and  Lemma \ref{laptev_weidl2}
in the last inequality.

For $t\in [-1,1]$,
let $E_1(\omega,t)$
be the lowest eigenvalue of
$\left(-\Delta_N^k + t V_\omega\right)/2$,
and $\phi(\omega, t)$ the normalized eigenfunction
corresponding to $E_1(\omega,t)$.
In particular, $E_1(\omega, 0)=0$ and
$\phi(\omega, 0) = {1}/{\sqrt{|Q_k|}}$.
Since $E_1(\omega,0)$ is a simple eigenvalue,
we can assume $E_1(\omega,t)$ and $\phi(\omega,t)$
are differentiable at $t=0$
by the analytic perturbation theory
\cite{K}.
By the Feynman-Hellmann theorem
\cite[Theorem 4.1.29]{S1}, we have
\begin{equation}
\label{prmain2}
 E_1'(\omega,0)
 = \frac{1}{2} (V_\omega \phi(\omega,0), \phi(\omega,0))
 = \frac{1}{2|Q_k|}\int_{Q_k} V_\omega(z) dx,
\end{equation}
where $'$ denotes the derivative
with respect to $t$.
For $n \in \mathbb{Z}^2$,
put
\begin{displaymath}
\beta_n(\omega)=
\frac{1}{2} \int_{n+Q_0}V_\omega(z)dx.
\end{displaymath}
Then, the random variables 
$\{\beta_n\}_{n\in \mathbb{Z}^2}$
are i.i.d.\ and we have by (\ref{prmain2})
\begin{equation}
\label{derivative}
  E_1'(\omega,0)
 = \frac{1}{|Q_k|} \sum_{n\in Q_k \cap \mathbb{Z}^2}
   \beta_n(\omega).
\end{equation}

\begin{lemma}
Let 
\begin{displaymath}
 s_0 = -\frac{1}{2} \log \mathbf{E}[\exp(-\beta_0(\omega))].
\end{displaymath}
 Then, we have
\begin{equation}
 \label{prmain3}
\mathbf{P}\left\{
\frac{1}{|Q_k|} \sum_{n\in Q_k \cap \mathbb{Z}^2}
   \beta_n(\omega) \leq s_0
\right\}
\leq e^{- s_0 |Q_k|}.
\end{equation}
\end{lemma}
\begin{proof}
By (iv) of Assumption \ref{assumptions},
the random variable
$\beta_n(\omega)$ takes positive value
with positive probability.
Thus $s_0$ is positive and
\begin{eqnarray*}
(\mbox{l.h.s. of (\ref{prmain3})})
&=&
\mathbf{P}\left\{
\sum_{n\in Q_k \cap \mathbb{Z}^2}
   \beta_n(\omega)
 \leq s_0 |Q_k|
\right\}\\ 
&\leq&
\mathbf{E}
\left[
\exp\left(
  s_0 |Q_k| 
- 
\sum_{n\in Q_k \cap \mathbb{Z}^2}
   \beta_n(\omega)
\right)
\right]\\
&=&
e^{s_0 |Q_k|}
\left(
\mathbf{E}
[e^{-\beta_0(\omega)}]
\right)^{|Q_k|}
= e^{-s_0|Q_k|},
\end{eqnarray*}
where we used the independence of the random variables
$\{\beta_n\}_{n\in \mathbb{Z}^2}$.
\end{proof}

The interval between the lowest two eigenvalues
of $-\Delta_N^k/2$ is 
$\left({\pi}/(2k+1)\right)^2/2$.
Notice that $\|V_\omega/2\|_\infty\leq 1/2$.
By the analytic perturbation theory
\cite[Theorem 4.1.30]{S1},
we see that 
the eigenvalue $E_1(\omega,z)$ can be extended 
analytically in the region 
$\{z\in \mathbb{C}\mid |z| < R/2\}$,
where
\begin{displaymath}
R = \left(\frac{\pi}{2k+1}\right)^2.
\end{displaymath}
Moreover, 
\begin{equation}
\label{prmain4}
 |E_1(\omega, z)|< 
 \frac{1}{4}\left( \frac{\pi}{2k+1}\right)^2
\end{equation}
for $|z| < R/2$.

\begin{lemma}
 \label{prmain5}
Put $C_7=2/\pi^2$.
Then
\begin{equation}
 \left|
  E_1(\omega, t) - t E_1'(\omega, 0)
\right| \leq 
C_7 |Q_k| t^2
\end{equation}
for every real $t$
with $|t|< {R}/{2}$.
\end{lemma}
\begin{proof}
Since $E_1(\omega, 0)=0$,
$E_1(\omega, t) - t E_1'(\omega, 0)$
equals to 
the remainder term
$(E''(\omega,\xi)/2) t^2$
($|\xi|<|t|$)
in the Taylor expansion.
By the Cauchy integral formula
and (\ref{prmain4}),
we have
for $|t|<{R}/{2}$ and small $\epsilon>0$
\begin{eqnarray*}
  \left|\frac{E_1''(\omega,\xi)}{2}\right|
 &\leq& 
\frac{1}{2\pi}
\int_{|z|=(1-\epsilon)R}
 \left|
 \frac{E_1(\omega,z)}{(\xi - z)^3}
\right| |dz|\\
&\leq&
\frac{1}{4}\left(\frac{\pi}{2k+1}\right)^2
\frac{(1-\epsilon)R}{\left((1/2-\epsilon)R\right)^3}.
\end{eqnarray*}
Taking $\epsilon\rightarrow 0$, 
the right hand side converges to
\begin{displaymath}
 \frac{2}{R^2}\left(\frac{\pi}{2k+1}\right)^2
=
C_7|Q_k|.
\end{displaymath}
Thus we have the conclusion.
\end{proof}

Let $b$ be a small positive number, which will be determined later.
Assume the event
\begin{displaymath}
 E_1(\omega,1) \leq  \frac{b}{(2k+1)^2}
\end{displaymath}
occurs. Since $E_1(\omega,t)$ is monotone non-decreasing
with respect to $t$,
we have $ E_1(\omega,t) \leq  {b}/{(2k+1)^2}$
for $|t|\leq 1$.
By Lemma \ref{prmain5}, we have
\begin{eqnarray}
 E_1'(\omega,0 ) 
&\leq&
\frac{1}{t} E_1(\omega, t)
+ C_7 (2k+1)^2 t\nonumber\\
&\leq&
\frac{b}{t} \frac{1}{(2k+1)^2}
+ C_7 (2k+1)^2 t
\label{prmain7}
\end{eqnarray}
for $|t|<{R}/{2}$.
The right hand side of (\ref{prmain7})
takes the minimal value $2\sqrt{b C_7}$
at $t = t_0 = \sqrt{b/C_7}(2k+1)^{-2}$.
Take $b$ so small that
\begin{displaymath}
 2 \sqrt{b C_7}\leq s_0,\quad
t_0 = \sqrt{\frac{b}{C_7}}(2k+1)^{-2}<\frac{R}{2}=\frac{\pi^2}{2}(2k+1)^{-2}.
\end{displaymath}
Then we have from (\ref{prmain7}) with $t=t_0$
\begin{displaymath}
 E_1'(\omega, 0) \leq  s_0.
\end{displaymath}
This inequality, (\ref{derivative}) and (\ref{prmain3})
implies
\begin{eqnarray*}
 \mathbf{P}\left\{
 E_1(\omega, 1) \leq \frac{b}{(2k+1)^2}
\right\}
 &\leq&
\mathbf{P}\{ E_1'(\omega, 0) \leq s_0\}\\
&\leq&
 e^{-s_0 |Q_k|}.
\end{eqnarray*}
This inequality and (\ref{prmain1.5})
implies the conclusion of Theorem \ref{theorem_main}.
\vspace{2mm}

\textbf{Acknowledgments.}
The authors thank to the referee for many helpful comments.
The work of T.\ M.\ is partially 
supported by JSPS grant Kiban C--26400165.
The work of Y.\ N.\ is partially
supported by JSPS grant Kiban C--15K04960.

%
%
%
%

\end{document}